\documentclass[a4paper,conference]{IEEEtran}
\IEEEoverridecommandlockouts
\usepackage{cite}
\usepackage{amsmath,amssymb,amsfonts}
\usepackage{algorithmic}
\usepackage{algorithm}
\usepackage{graphicx}
\usepackage{textcomp}
\usepackage{xcolor}
\usepackage{amsthm}
\usepackage{mathtools,flexisym}
\usepackage{bm}
\usepackage{enumitem}
\usepackage{graphicx}
 \usepackage{verbatim}
 \usepackage{subfigure}

\newtheorem{Example}{Example}
\usepackage{amsthm}
\def\BibTeX{{\rm B\kern-.05em{\sc i\kern-.025em b}\kern-.08em
    T\kern-.1667em\lower.7ex\hbox{E}\kern-.125emX}}

\begin{document}

\title{Age of Information Guaranteed Scheduling for Asynchronous Status Updates in Collaborative Perception}

\author{\IEEEauthorblockN{Lehan~Wang$^*$, Jingzhou~Sun$^*$, Yuxuan~Sun$^\dag$, Sheng~Zhou$^*$, Zhisheng~Niu$^*$}
	
	\IEEEauthorblockA{$^*$Beijing National Research Center for Information Science and Technology\\
		Department of Electronic Engineering, Tsinghua University, Beijing 100084, P.R. China\\
		$^\dag$School of Electronic and Information Engineering, Beijing Jiaotong University, Beijing 100044, China\\
		\{wang-lh19, sunjz18\}@mails.tsinghua.edu.cn, yxsun@bjtu.edu.cn, \{sheng.zhou, niuzhs\}@tsinghua.edu.cn}}

\maketitle

\begin{abstract}
We consider collaborative perception (CP) systems where a fusion center monitors various regions by multiple sources. The center has different age of information (AoI) constraints for different regions. Multi-view sensing data for a region generated by sources can be fused by the center for a reliable representation of the region. To ensure accurate perception, differences between generation time of asynchronous status updates for CP fusion should not exceed a certain threshold. An algorithm named scheduling for CP with asynchronous status updates (SCPA) is proposed to minimize the number of required channels and subject to AoI constraints with asynchronous status updates. SCPA first identifies a set of sources that can satisfy the constraints with minimum updating rates. It then chooses scheduling intervals and offsets for the sources such that the number of required channels is optimized. According to numerical results, the number of channels required by SCPA can reach only 12\% more than a derived lower bound. 
\end{abstract}

\section{Introduction}
The rapid development of wireless communication technologies enables information exchange among massive smart terminals, which is the cornerstone of future Internet of Things (IoT) applications such as autonomous driving (AD). Timely and reliable status information collected from terminals is often crucial to such smart applications. For example, an autonomous vehicle requires status information such as the current speed and location of nearby traffic participants to make decisions. To this end, road-side
infrastructure sensors are required to provide their sensing data. However, the vehicle may fail to obtain reliable information with the aid of a single sensor when there are objects out of its sensing range or occluded \cite{jia2023mass}. In this case, we need sensors to fulfill collaborative perception (CP). Data provided by the sensors will be fused at a fusion center and the multi-view perception results will then be shared with nearby vehicles \cite{Arnold2022coop}.  
One way is to let all nearby sensors transmit their data to the center. However, such an approach requires a large amount of wireless resources and is unnecessary. Ref.~\cite{jia2023mass,liu2020when2com,liu2020who2com} show that the multi-view perceptual results produced by a subset of nearby sensors are enough to help the vehicle drive safely. One reason is that sensing ranges of the sensors partially overlap. In addition, 
each fusion requires 
sensing data generated within a certain interval due to dynamic driving environments. In \cite{lei2022latency}, the authors claim that directly fusing asynchronous information generated at different time may lead to even worse performance than single-sensor perception. The authors then propose a fusion algorithm obtaining CP gain if the difference between generation time of data for fusion does not exceed a certain threshold. Meanwhile, since status information changes rapidly, the center also requires timely sensing data for the safety of the vehicle.  

In this paper, we introduce Age of Information (AoI), a widely used metric proposed in \cite{kaul2012real}, to measure information timeliness. AoI is defined as the time elapsed since the generation time of the latest received status information. Most existing work of AoI focus on minimizing average age or age-based functions \cite{yates2021age}, with limited attention paid to the harm of extreme AoI in time-critical applications such as AD \cite{abdel2019optimized}. Additionally, AoI requires adjustment such that we can employ the metric to measure performance of CP systems with asynchronous status updates. Therefore, we focus on a novel scheduling problem with AoI guarantee. In the model, there is a fusion center monitoring various regions by multiple sources. For reliable information of a region, the center may require status updates from several sources. We are interested in how to schedule sources such that the number of required channels is minimized while AoI of each region does not exceed a certain threshold. Such thresholds are hereinafter referred to as AoI constraints.

Ref.~\cite{li2022scheduling,li2021scheduling,liu2021aion,wang2022grouping,he2019joint} are most relevant work. In \cite{li2022scheduling}, the authors identify conditions for AoI constraints such that a feasible scheduling policy can always be constructed given a single reliable channel. In \cite{li2021scheduling}, the authors consider a problem similar to that in \cite{li2022scheduling} while the channel is assumed to be unreliable. Ref.~\cite{liu2021aion,wang2022grouping} try to figure out the minimum number of channels required under AoI constraints and the corresponding scheduling policies.  However, the above papers do not take collaboration among sources into consideration. In \cite{he2019joint}, the authors consider a wireless camera network with multi-view image processing and focus on minimizing the maximum peak age of all monitored scenes. The model is similar to ours but each image processing requires synchronous photos from all cameras monitoring the same scene. In our model, fusion algorithms can handle information generated at different time by only a part of sources.

The contributions of this paper are summarized as follows:

\begin{itemize}
	\item A novel scheduling problem under hard AoI constraints in CP systems with asynchronous status updates is formulated. We adjust AoI such that the metric can be used in the considered systems. This adjusted metric can be used in other systems where information for decision-making is derived based on asynchronous status updates from multiple sources. 
	\item In the model, status updates of each source can change AoI of multiple regions and the AoI of each region depends on several sources. 
	 A novel lower bound of the minimum number of required channels is derived based on the complicated topology of sources and regions.
	\item A special type of scheduling policy named homogeneous scheduling is considered. Under such policies, intervals between any two consecutive scheduling decisions of each source are the same.  
	We propose an algorithm named scheduling for CP with asynchronous status updates (SCPA) to construct such scheduling policies and optimize the number of required channels. The first step of SCPA is to identify a set of sources that can satisfy AoI constraints with the minimum updating rate. It then designs a policy for the sources based on a derived condition of scheduling intervals,  such that AoI constraints can be satisfied. 
	\item 
	According to numerical results, SCPA satisfies AoI constraints for the simulated instances. Additionally, the number of channels required by SCPA can reach 12\% more than the lower bound.
\end{itemize}

The rest of the paper is organized as follows. The system model is shown in Section II. Problem formulation and a derived lower bound of the number of required channels are presented in Section III. Details of SCPA are shown in Section IV. Finally, Section V shows the numerical results and Section VI concludes the paper.

\section{System Model}

\begin{figure}[t]
	
	\centering
	\includegraphics[scale=0.15]{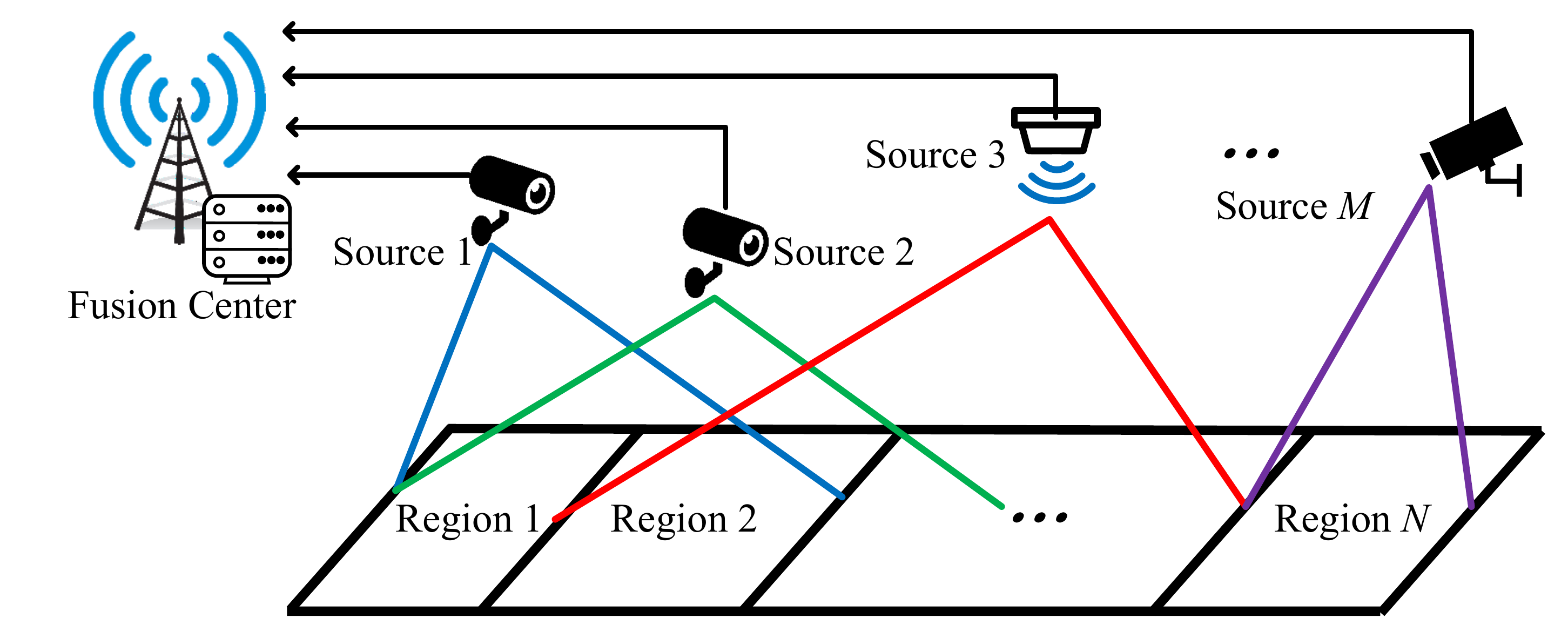}
	\vspace{-0.1in}
	\caption{System model.}
	\label{system}
\end{figure}

As shown in Fig.~\ref{system}, we consider a time-slotted CP system where there is a fusion center collecting information of $N$ regions. There are $M$ sources (sensors or cameras) transmitting their status updates to the center through $K$ reliable channels. The set of sources monitoring the $n$-th region is denoted by $\bm{\mathcal{K}}_n$, which will not change over time. Due to different positions and orientations of sources,
some of them may have a view of multiple adjacent regions. 
Measurements performed by such sources will include information of several regions.

At the beginning of each time slot, the fusion center will schedule a set of sources to generate status information and transmit the updates to the center. Indicator $U_m(t)\in\{0,1\}$ denotes the scheduling decision of source $m$. If $U_m(t)=1$, the $m$-th source is scheduled at time $t$. Otherwise, the source remains idle. Each channel can let at most one source transmit at a time, then $\sum_{m=1}^MU_m(t)\leq K$ holds for all slots. The status information transmitted at $t$ is received at the end of the slot and all received information will be kept by the center.

Assume sources in each $\bm{\mathcal{K}}_n$ are heterogeneous in terms of the sensing range. 
We assume the center can reconstruct information of the region with a status update from any source in set $\bm{\mathcal{F}}_n\neq\emptyset$. 
In $\overline{\bm{\mathcal{F}}}_n=\bm{\mathcal{K}}_n\setminus\bm{\mathcal{F}}_n$, there exists $s_n\in\mathbb{Z}$ different combinations of sources whose status updates can be fused by the center for comprehensive information of the region, denoted by: $\{\bm{k}_{n,1},\bm{k}_{n,2},\cdots,\bm{k}_{n,s_n}\}$. Each combination includes at least 2 different sources, namely $|\bm{k}_{n,j}|\geq2,\forall j$. 
Meanwhile, $\bm{k}_{n,i}$ is not a subset of $\bm{k}_{n,j}$ 
 for any $i\neq j$.

	Fusion algorithms employed by the center can tolerate asynchronous status information from multiple sources to a certain extent. Assume that as long as the gap between the generation time of the input status updates does not exceed $T_n\in\mathbb{Z}^{+}$, the perception accuracy of the region can be satisfied. Otherwise, the fusion will fail. Additionally, we assume that the center possesses massive computing resources. Thus the time consumption of fusions can be ignored.

Age of information (AoI) is introduced to measure information timeliness. Let $A_n(t)\in\mathbb{Z}^{+}$ denote age of the $n$-th region. Since only successful fusions bring the center reliable information of the region, then:
\begin{equation}\label{aoi}
\begin{aligned}
	A_n(t+1)=
	\begin{cases}
		1, &  {\sum_{m\in\bm{\mathcal{F}}_n}U_m(t)\geq1}\\
		1, &  \text{There exists $j$ such that} \\ & \sum_{m\in \bm{k}_{n,j}} U_m(t)\geq1 \text{ and}\\ & t-\min_{m\in\bm{k}_{n,j}}g(m,t)\leq T_n\\
		A_n(t)+1,  & {\text{Otherwise,}}
	\end{cases} 
\end{aligned}
\end{equation}
where $g(m,t)$ is the generation time of the latest received update from source $m$ at slot $t$. In the first case of \eqref{aoi}, scheduling any source in $\bm{\mathcal{F}}_n$ at $t$ leads to a reduction of age in $t+1$. In the second case, a successful fusion of status updates generated by any combination $\bm{k}_{n,j}$ leads to an AoI reduction. We assume the center will start a fusion for $\bm{k}_{n,j}$ when at least a status update from $\bm{k}_{n,j}$ is received at the end of the last slot and the latest received status updates from all sources in $\bm{k}_{n,j}$ are generated in an interval of length $T_n$. 
Otherwise, information of region $n$ possessed by the center will age. 

Fig.~\ref{aoi_mosc} shows an example of the age evolution of region $1$ when $A_1(t)=1, \bm{\mathcal{F}}_1$ includes a single source $A$, $s_1=1$, $\bm{k_{1,1}}=\{B,C\}$ and $T_1=1$. We schedule source $A$ twice at time $2$ and $4$. $B$ is scheduled at time $3$ and $8$ while $C$ is scheduled at time $6$ and $9$. In the figure, we can see that each scheduling of source $A$ leads to AoI reduction in the next slot. Additionally, AoI reduction will occur when a fusion is done successfully and $A_1(10)=1$. $A_1(7)>1, A_1(9)>1$ because the latest status updates from $\{B,C\}$ do not satisfy the conditions of the second case in \eqref{aoi}.

For the information of region $n$, the center requires that its age does not exceed $d_n\in\mathbb{Z}^{+}$ over time, namely $A_n(t)\leq d_n, \forall t$. Since it is not reasonable to use too outdated status information for fusion, we assume $T_n<d_n$. Without loss of generality, the regions are sorted such that $d_1\leq d_2 \leq d_3 \cdots\leq d_N$. We are interested in the minimum number of channels to satisfy AoI constraints and a corresponding scheduling policy $\bm{\pi}$. In this paper, a certain scheduling policy is a matrix with $T$ rows and $M$ columns, where $T$ denotes the interval of interest and $T\to\infty$ in this work. The $t$-th row of $\bm{\pi}$ includes scheduling decisions of all sources, namely $\bm{\pi}(t)=[U_1(t),\cdots,U_m(t), \cdots, U_M(t)]$.

\begin{figure}[t]
	\centering
	\includegraphics[scale=0.35]{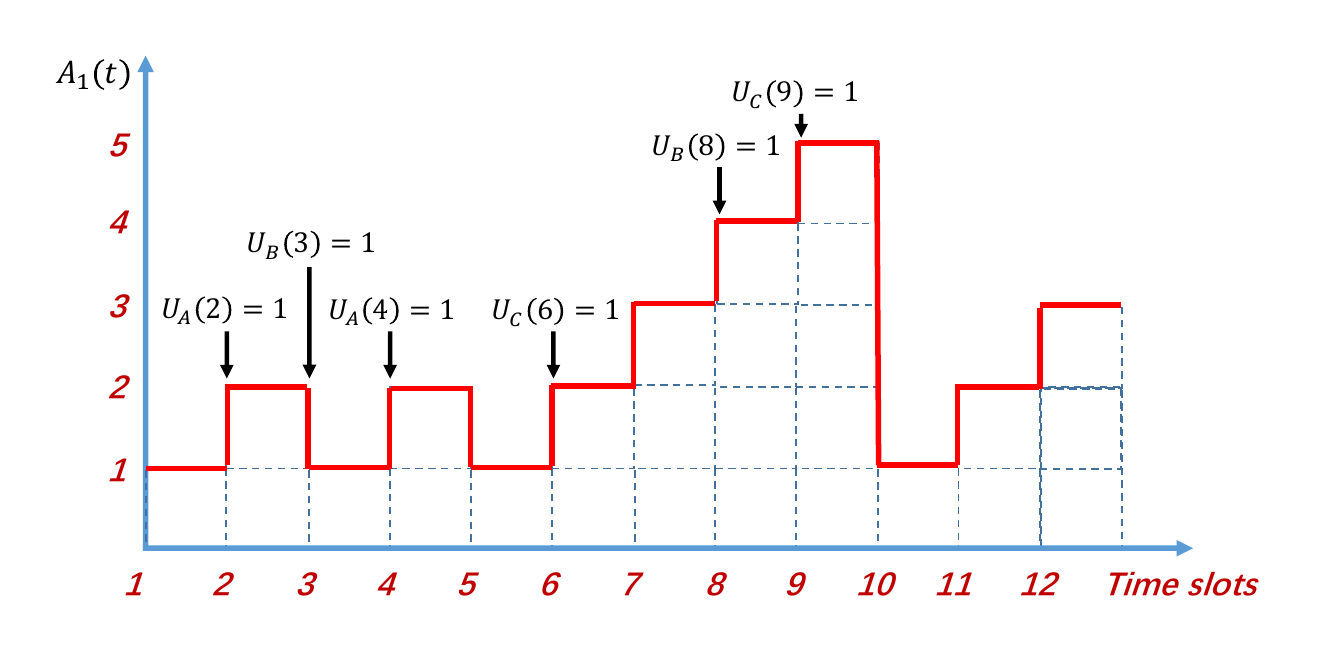}
	\vspace{-0.2in}
	\caption{An example of evolution of AoI.}
	\label{aoi_mosc}
\end{figure}

 \addtolength{\topmargin}{+0.069cm}
\section{Problem Formulation and Analysis}

The problem can be formulated as follow: 
 \begin{equation}\label{problem}
	\begin{split}
		\min_{\bm{\pi}} \,\,& K,\\
		s.t.\,\,& A_n(t)\leq d_n,\forall n=1,\cdots,N, \forall t=1,2,\cdots.\\
	\end{split}
\end{equation}

A lower bound of this problem is presented in the following lemma. 

\newtheorem{lemma}{Lemma}
\begin{lemma} \label{lbub}
	Given AoI constraints $\bm{d} =
	[d_1, d_2, \cdots ,d_N]$, network setup $\{\bm{\mathcal{K}}_1, \cdots, \bm{\mathcal{K}}_N\}$ and requirements for successful fusion $\{T_1,\cdots,T_N\}$, we have 
	\begin{equation}
		K^{\star}\geq\left\lceil\sum_{m=1}^{M}{l_m^{\star}}\right\rceil,
	\end{equation}
	where $K^{\star}$ is the optimal number of required channels, $\bm{l}^{\star}=\{l_1^{\star},\cdots,l_M^{\star}\}$ is the solution to the following problem.
	\begin{equation}\label{lb}
		\begin{split}
			\min_{\bm{l}} & \left\lceil{ \sum_{m=1}^{M}l_m}\right\rceil
			\\
			s.t.&  \left\{\begin{array}{lc}
				\sum_{j=1}^{s_n}\sum_{m\in\bm{k}_{n,j}}l_m+\sum_{m\in\bm{\mathcal{F}}_n}l_m\geq \frac{1}{d_n}, \forall n\\
				0\leq l_m\leq1, \forall m\\
			\end{array}\right.
		\end{split}
	\end{equation}
\end{lemma}

\begin{proof}
	To prove this lemma, we first find out the minimum channel resources required by region $n$. If we only employ status updates from sources in $\bm{\mathcal{F}}_n$ to satisfy $d_n$, then there should be at least one transmission from any source in $\bm{\mathcal{F}}_n$ within $d_n$ consecutive time slots. Let  $l_m=\lim_{T\to\infty}\frac{1}{T}\sum_{t=1}^TU_m(t)\leq1$ denote the updating rate of source $m$. Therefore, $\sum_{m\in\bm{\mathcal{F}}_n}l_m\geq\frac{1}{d_n}$, which means at least $\frac{1}{d_n}$ of the time slots of a channel is required by the region in this case.
	
	Next, we discuss the minimum required channel resources if we only schedule sources in $\bm{{k}}_{n,j}$ to satisfy $d_n$. In this case, scheduling any source in $\bm{{k}}_{n,j}$ may lead to AoI reduction in the next slot because the fusion algorithm can handle status updates generated at different time. Therefore, $\sum_{m\in\bm{k}_{n,j}}l_m\geq\frac{1}{d_n}$. Since there are $s_n$ combinations of source leading to successful fusions, we require:  
	\begin{equation}\label{lb2}
		\sum_{j=1}^{s_n}\sum_{m\in\bm{k}_{n,j}}l_m+\sum_{m\in\bm{\mathcal{F}}_n}l_m\geq \frac{1}{d_n}, \forall n,
	\end{equation}
which is the constraint in \eqref{lb}. Since $l_m$ denotes the minimum channel resources that should be assigned to source $m$, $\lceil{ \sum_{m=1}^{M}l_m}\rceil$ is the lower bound of channels required by $\bm{l}=\{l_1,\cdots,l_M\}$. Thus the proof is completed.  
\end{proof}

 A challenge of solving \eqref{problem} lies in determining the infinite elements of $\bm{\pi}$. Therefore, we focus on a special type of scheduling policy named homogeneous scheduling. The definition is shown below. 

\newtheorem{definition}{Definition}
\begin{definition} \label{cyclic scheduler}
	A scheduling policy is homogeneous if and only if all the sources are scheduled with the same interval. 
\end{definition}

Admittedly, considering only such homogeneous scheduling policies may lead to sub-optimal solutions. But it does make it easier for us to solve the problem, and numerical results in section V validate this simplification. In this case, we only need to determine scheduling interval and offset, denoted by positive integers $c_m$ and $o_m$, for each source. Here, $o_m$ is the moment when source $m$ is scheduled for the first time and we assume $o_m\leq c_m$. According to the definition, $o_m+ic_m$ is the moment when source $m$ is scheduled for the $(i+1)$-th time, $i\in\mathbb{N}$.

\newtheorem{thm}{Theorem}

\section{Scheduling for CP with Asynchronous \\ status updates (SCPA)}

To solve eq.~\eqref{problem}, an algorithm named scheduling for CP with asynchronous status updates (SCPA) is proposed. SCPA includes two steps:
\begin{itemize}
\item determine active sources which can satisfy AoI constraints with the minimum updating rates.
\item establish a homogeneous scheduling policy for the sources by carefully choosing scheduling intervals and offsets.
\end{itemize}

Details of SCPA are shown in the following subsections.

\subsection{Determining Active Sources}

In our model, the sensing range of several sources may partially overlap, and thus there may be no need to activate and schedule all sources. To optimize the number of required channels, we first determine active sources which can satisfy AoI constraints with the minimum updating rates. 
For the simplicity of constructing a corresponding scheduling policy, we only activate one source in $\bm{\mathcal{F}}_n$ or a combination of sources for each region. Additionally, we require updating rates of sources activated for region $n$ are no less than $\frac{1}{d_n}$ to ensure $A_n(t)\leq d_n$. 
Based on the analysis above, active sources can be determined by solving the following problem:
\begin{equation}\label{active source}
	\begin{split}
		\min_{\bm{l},\bm{Q}} & { \sum_{m=1}^{M}{l_m}}
		\\
		s.t.&  \left\{\begin{array}{lc}
			\sum_{m\in\bm{\mathcal{F}}_n}Q_{m,n}+\\\quad\quad\sum_{j=1}^{s_n}\mathbb{I}(\min_{m\in\bm{k}_{n,j}}Q_{m,n}\geq1)=1, \forall n\\
						\frac{Q_{m,n}}{d_n}\leq l_m\leq1, \forall m,\\
		\end{array}\right.
	\end{split}
\end{equation}
where indicator function $\mathbb{I}(x\geq y)=1$ if $x\geq y$ holds, otherwise $\mathbb{I}(x\geq  y)=0$.  
 $\bm{Q}$ is a matrix with $M$ rows and $N$ columns. The value of all elements in $\bm{Q}$ can only be $0$ or $1$. $Q_{m,n}$ denotes the $n$-th element on the $m$-th row of $\bm{Q}$. If $Q_{m,n}=1$, $l_m\geq\frac{1}{d_n}$. If $Q_{m,n}=0$, $l_m$ is independent with $d_n$. Let $\bm{\mathcal{M}}_{n}$ and $\bm{l}^{\prime}=\{l_1^{\prime},\cdots,l_M^{\prime}\}$ denote sources activated for region $n$ and a solution to \eqref{active source}, respectively. If source $m$ belongs to $\bm{\mathcal{M}}_{n}$, then $Q_{m,n}=1, l_m^{\prime}>0$ and $m\in\bm{\mathcal{F}}_n$ or $m$ belongs to the set $\bm{k}_{n,j}$ such that $\min_{m\in\bm{k}_{n,j}}Q_{m,n}\geq1$. The maximum scheduling interval of active source $m$ is $\frac{1}{l_m^{\prime}}=\min_{n\in\bm{\mathcal{N}}_m}d_n$, where $\bm{\mathcal{N}}_m=\{n|m\in\bm{\mathcal{M}}_{n}\}$.

\subsection{Conditions for Policies Satisfying AoI Constraints}

After determining active sources and their maximum scheduling intervals, the next step is to determine $o_m$ and $c_m$ for each source such that the AoI constraints are satisfied. First of all, if $\sum_{m=1}^MQ_{m,n}=1$ holds for all regions, there is no collaboration between the active sources. Thus policies   scheduling each active source at intervals less than or equal to  $\frac{1}{l_m^{\prime}}$ time slots can satisfy the constraints. Construction of such a scheduling policy has already been studied in \cite{liu2021aion,wang2022grouping} and a near-optimal scheduling policy named two-step grouping algorithm (TGA) is proposed in \cite{wang2022grouping}.

If we activate a combination of sources, denoted by $\bm{k}_{n,j}$, for region $n$, then we have $\frac{1}{l_{z_n}^{\prime}}\leq d_n$. Here $z_n$ is the source with the lowest updating rate $l_{z_n}^{\prime}$ in $\bm{k}_{n,j}$. Since we focus on homogeneous scheduling, then one way to minimize the required number of channels is scheduling sources in $\bm{k}_{n,j}$ with the maximum interval, namely letting $c_{m}=\frac{1}{l_{m}^{\prime}}(\forall m\in\bm{k}_{n,j})$. Given scheduling intervals, we can satisfy AoI constraints $d_n$ by carefully designing offsets $o_m$ such that every scheduling of source $z_n$ leads to an AoI reduction in the next slot. This is because the interval between any two consecutive AoI reductions will not exceed $c_{z_n}$.

Since each scheduling of $z_n$ leads to AoI reduction in the next slot, source $z_n$ should possess the largest offsets over $\bm{k}_{n,j}$. Meanwhile, the required offsets of sources in $\bm{k}_{n,j}$ should satisfy the following inequalities for each $i\in\mathbb{N}$:
\begin{equation} \label{offsets_general scheduling interval}
 \tau_i-o_{m}\;\mathrm{mod}\; c_{m}\leq T_n, \forall m\in\bm{k}_{n,j}\setminus z_n,
\end{equation} 
where function $a\;\mathrm{mod}\; b$ equals the remainder when $a$ is divided by $b$ and $\tau_i=ic_{z_n}+o_{z_n}$ is the moment when $z_n$ is scheduled for the $(i+1)$-th time. Then the last time $m$ was scheduled before $\tau_i$ is $\lfloor\frac{\tau_i-o_m}{c_m}\rfloor c_m+o_m$. Therefore, LHS of \eqref{offsets_general scheduling interval} is the gap between $\lfloor\frac{\tau_i-o_m}{c_m}\rfloor c_m+o_m$ and $\tau_i$. If the gap does not exceed $T_n$ for all sources in $\bm{k}_{n,j}\setminus z_n$ and all scheduling indices $i$, $A_n(\tau_i+1)=1$.

Note that there may not be offsets such that \eqref{offsets_general scheduling interval} holds. For example, consider region $n$ with $d_n=4$ and the combination of sources activated for the region is $\{A,B\}$. If $c_A=4$, $c_B=3$ and $T_n=1$, then no matter how we choose $o_A$ and $o_B$, the interval between two consecutive AoI reductions will be more than $d_n$ time slots. The following lemma presents conditions under which there exist offsets of $\bm{k}_{n,j}$ such that \eqref{offsets_general scheduling interval} hold.

\begin{lemma} \label{condition_offset}
	There exist $o_m(m\in\bm{k}_{n,j})$ such that \eqref{offsets_general scheduling interval} holds if and only if $T_n\geq\max_{m\in\bm{k}_{n,j}\setminus z_n}\{c_m-\gcd(c_m,c_{z_n})\}$, where $\gcd(c_m,c_{z_n})$ is the greatest common divisor of $c_m$ and $c_{z_n}$.
\end{lemma}

\begin{proof}
First of all, $(\tau_i-o_{m}\;\mathrm{mod}\; c_{m})\leq c_m-1$ holds due to the definition of remainder.  
Since $c_m$ is divisible by $\gcd(c_{z_n},c_m)$, then we have $(a\;\mathrm{mod}\;c_m)\;\mathrm{mod}\;\gcd(c_{z_n},c_m)=a\;\mathrm{mod}\;\gcd(c_{z_n},c_m)$. Thus:
\begin{equation}\label{mod}
	\begin{aligned}
		&(\tau_i-o_{m}\;\mathrm{mod}\; c_{m})\;\mathrm{mod}\;\gcd(c_{z_n},c_m)\\ =&(ic_{z_n}+o_{z_n}-o_{m})\;\mathrm{mod}\; \gcd(c_{z_n},c_m)\\
		\overset{(a)}{=}&(ic_{z_n}\;\mathrm{mod}\;\gcd(c_{z_n},c_m) +\\&(o_{{z_n}}-o_m)\;\mathrm{mod}\;\gcd(c_{z_n},c_m) )\;\mathrm{mod}\; \gcd(c_{z_n},c_m)\\
		=&(o_{{z_n}}-o_m)\;\mathrm{mod}\;\gcd(c_{z_n},c_m),
	\end{aligned}
\end{equation}
where equality (a) holds because $(a+b)\;\mathrm{mod}\; n=(a\;\mathrm{mod}\;n +b\;\mathrm{mod}\;n )\;\mathrm{mod}\; n$. 
Due to eq.~\eqref{mod} and the fact that $(\tau_i-o_{m}\;\mathrm{mod}\; c_{m})$ can not exceed $c_m-1$, the maximum achievable value of $(\tau_i-o_{m}\;\mathrm{mod}\; c_{m})$ is $ c_m-\gcd(c_{z_n},c_m)+(o_{z_n}-o_m)\;\mathrm{mod}\; \gcd(c_{z_n},c_m)$.   
We then prove that $\tau_i-o_m\;\mathrm{mod}\;c_m$ can actually achieve $c_m-\gcd(c_{z_n},c_m)+(o_{z_n}-o_m)\;\mathrm{mod}\; \gcd(c_{z_n},c_m).$ Let $a_i=(\tau_i-o_{m}\;\mathrm{mod}\; c_{m})$. Then there exists $y_i$ such that $ic_{z_n}-y_ic_{m}=a_i-(o_{z_n}-o_m)$. If $a_j$ has the same value as $a_i$, then the absolute value of $j-i$ is divisible by $\frac{c_m}{\gcd(c_{z_n},c_m)}$ according to linear Diophantine equations. Therefore, for any $i\in\mathbb{N}$, elements in the set $\{a_i,a_{i+1},\cdots,a_{i+\frac{c_m}{\gcd(c_{z_n},c_m)}-1}\}$ have different values. Since eq.~\eqref{mod} holds for all $i$ and $a_i\leq c_m-1$, $a_i$ has $\frac{c_m}{\gcd(c_{z_n},c_m)}$ different values, denoted by set $\{(o_{z_n}-o_m)\;\mathrm{mod}\; \gcd(c_{z_n},c_m), (o_{z_n}-o_m)\;\mathrm{mod}\; \gcd(c_{z_n},c_m)+\gcd(c_{z_n},c_m),\cdots,c_m-\gcd(c_{z_n},c_m)+(o_{z_n}-o_m)\;\mathrm{mod}\; \gcd(c_{z_n},c_m)\}$. Thus there must exists $i$ such that $a_i=c_m-\gcd(c_{z_n},c_m)+(o_{z_n}-o_m)\;\mathrm{mod}\; \gcd(c_{z_n},c_m)$. Since $(o_{z_n}-o_m)\;\mathrm{mod}\; \gcd(c_{z_n},c_m)\geq0$, $T_n\geq c_m-\gcd(c_m,c_{z_n}$ can guarantee the existence of a pair of $o_m,o_{z_n}$ satisfying $\tau_i-o_{m}\;\mathrm{mod}\; c_{m}\leq T_n$, $\forall i$. To ensure there exist offsets $o_m(m\in\bm{k}_{n,j})$ such that \eqref{offsets_general scheduling interval} holds, $T_n$ should be no less than $\max_{m\in\bm{k}_{n,j}\setminus z_n}\{c_m-\gcd(c_m,c_{z_n})\}$.
\end{proof}

\subsection{Basis for Determining Scheduling Intervals}

Lemma \ref{condition_offset} tells us how to construct scheduling policies given relatively large $T_n$. However, if $T_n$ is small, directly setting $c_m=\frac{1}{l_m^{\prime}}$ may not be able to satisfy AoI constraints. In this case, we have to adjust the scheduling intervals such that conditions in Lemma \ref{condition_offset} can be satisfied. More notably, setting $c_m$ as the maximum value $\frac{1}{l_m^{\prime}}$ may not be a good choice even from the perspective of reducing the number of required channels. The following example shows that adjusting scheduling intervals can not only help us to satisfy the constraints but also reduce the required wireless resources. 

\begin{Example} \label{example 2}
	Consider region $n$ with $d_n=4$ and $T_n=2$. Active source of the region are $\bm{\mathcal{M}}_n=\{A,B\}$ and $\frac{1}{l_{A}^{\prime}}=4$, $\frac{1}{l_{B}^{\prime}}=3$. If $c_A=4$, $c_B=3$, $o_A$, $o_B$ satisfying \eqref{offsets_general scheduling interval} exist according to Lemma \ref{condition_offset}. However, no matter how we choose $o_A$, $o_B$, at least $2$ channels are required. When $o_A=2$, $o_B=1$, the scheduling policy is: 
	\vspace{-0.15in}
	\begin{center}
		\setlength{\tabcolsep}{1.3mm}{
			\begin{tabular}{l c c c c c c c c c c c c c} 		
				time slot&1&2&3&4&5&6&7&8&9&10&11&12\\
				channel 1 & B & $\square$ & $\square$ & B& $\square$ & $\square$ & B & $\square$ & $\square$ & B & $\square$ & $\square$ & $\cdots$	\\
				channel 2 &$\square$&A&$\square$ & $\square$ & $\square$ &A& $\square$ & $\square$ & $\square$ & A & $\square$ & $\square$& $\cdots$\\
		\end{tabular}}
	\end{center}
	If $c_B=2$, we only need a single channel. When $o_A=2$, $o_B=1$, the corresponding scheduling policy is shown below.
\begin{center}
	\setlength{\tabcolsep}{1.3mm}{
		\begin{tabular}{l c c c c c c c c c c c c c} 		
			time slot&1&2&3&4&5&6&7&8&9&10&11&12\\
			channel 1 & B & A & B & $\square$& B & A & B & $\square$ & B & A & B & $\square$ & $\cdots$	\\
	\end{tabular}}
\end{center}
\end{Example}

Example 1 can be explained as follows. According to the proof of Lemma \ref{condition_offset}, if $\gcd(c_{z_n},c_{m})=1$, then no matter how we choose $o_m$ and $o_{z_n}$, there exists $t$ such that $U_m(t)+U_{z_n}(t)=2$. This is because $\tau_i-o_{m}\;\mathrm{mod}\; c_{m}$ can always be $0$ no matter how we choose offsets. Thus we should avoid co-prime scheduling intervals for active sources. On the other hand, to satisfy AoI constraints with any $T_n>0$, $c_m$ and $c_{z_n}$ should be reduced to $c_m^{\prime}$ and $c_{z_n}^{\prime}$, such that $c_m^{\prime}$ is a divisor of $c_{z_n}^{\prime}$ for any $m\in\bm{k}_{n,j}\setminus z_n$ and  $c_m^{\prime}-\gcd(c_m^{\prime},c_{z_n}^{\prime})=0$.  

Since we need $c_{z_n}^{\prime}$ divisible by $c_{m}^{\prime}$ for sources in each activated $\bm{k}_{n,j}$, scheduling intervals of sources in each $\bm{\mathcal{V}}^{(i)} (0<i\leq b_g)$ should be designed jointly. Here, $\bm{\mathcal{V}}^{(i)}$ is the set of all vertices of $\bm{\mathcal{G}}^{(i)}$, which is the $i$-th (connected) component of undirected graph $\bm{\mathcal{G}}=\{\bm{\mathcal{V}},\bm{\mathcal{E}}\}$. $b_g\in\mathbb{Z}^{+}$ is the number of components of $\bm{\mathcal{G}}$. Set $\bm{\mathcal{V}}$ includes all activated sources, and the $i$-th vertex represents the $i$-th active source. Set $\bm{\mathcal{E}}$ includes all edges in the graph. If there is an edge between the $i$-th and $j$-th vertices, then the two corresponding sources are activated for at least a same region. The definition of component is:

\begin{definition} \label{component}
A component of an undirected graph is a subgraph of the graph. Any two vertices in a component can reach another via edges of the component, and a component has no edge connected to vertices out of the component.
\end{definition}

To ensure that $c_{z_n}^{\prime}$ is divisible by $c_{m}^{\prime}$ for sources in each activated $\bm{k}_{n,j}$, we require scheduling intervals of sources in each $\bm{\mathcal{V}}^{(i)}$ to be consecutively divisible (CD).  
\begin{definition} \label{CD}
	$\bm{x}=[x_1,x_2,\cdots,x_N]$ is consecutively divisible (CD) if $x_i\in \mathbb{Z}^{+}$, $\forall i\in\{1,\cdots,N\}$ and $x_i/x_{i-1}\in\mathbb{Z}^{+},\forall i\in\{2,\cdots,N\}$. The base of a CD vector is the minimum value of elements in $\bm{x}$.
\end{definition}

For each activated $\bm{k}_{n,j}$, there may exist more than one choice of CD scheduling intervals.  
The following example may shed light on how to choose CD intervals for minimizing $K$. 

\begin{Example} \label{example 2}
	Consider 5 regions with AoI constraints $[4,9,9,5,6]$. $T_n=1, \forall n=\{1,2,3\}$ and $T_4=T_5=2$. Sources activated for each region and the maximum scheduling intervals are: 1) $\bm{\mathcal{M}}_1=\{A,B\},\max_{m\in\bm{\mathcal{M}}_1}c_m\leq4$; 2) $\bm{\mathcal{M}}_2=\{C,D\},\max_{m\in\bm{\mathcal{M}}_2}c_m\leq9$; 3) $\bm{\mathcal{M}}_3=\{E\},c_E\leq9$; 4) $\bm{\mathcal{M}}_4=\{F,G,H\},\max_{m\in\bm{\mathcal{M}}_4}c_m\leq5$; 5) $\bm{\mathcal{M}}_5=\{H,I,J\},\max_{m\in\bm{\mathcal{M}}_5} c_m\leq6$. 

   $\bm{\mathcal{G}}$ has $4$ components: $\bm{\mathcal{V}}^{(1)}=\{A,B\}$, $\bm{\mathcal{V}}^{(2)}=\{C,D\}$, $\bm{\mathcal{V}}^{(3)}=\{E\}$, $\bm{\mathcal{V}}^{(4)}=\{F,G,H,I,J\}$.
 Any homogeneous policy satisfying the above constraints requires at least 2 channels. The following scheduling policy can reach the lower bound.
	
	\vspace{-0.06in}
	\begin{center}
		\setlength{\tabcolsep}{1.3mm}{
			\begin{tabular}{l c c c c c c c c c c c c c} 		
				time slot&1&2&3&4&5&6&7&8&9&10&11&12\\
				channel 1 & A & B & C & D& A & B & E & $\square$ & A & B & C & D & $\cdots$	\\
				channel 2 &F&G&H & I & J &F& G & H & I & J & F & G & $\cdots$\\
		\end{tabular}}
	\end{center}
    Under the policy, $c_A=c_B=4$, $c_C=c_D=c_E=8$, $c_F=c_G=c_H=c_I=c_J=5$. 
\end{Example}

 There are two points worth noting in Example 2: 
 \begin{itemize}
 \item The policy chooses not to assign CD scheduling intervals to each $\bm{\mathcal{V}}^{(i)}$ independently. For example, $c_C, c_D, c_E$ decreases from $9$ to $8$ for a CD structure with $c_A,c_B$. This is because we should avoid co-prime scheduling intervals. Note that the policy chooses not to decrease $c_F, c_G, c_H$ from 5 to 4 for a CD structure. This is because such an operation leads to a relatively great increase in the required wireless resources. In this case, it is more reasonable to allocate another channel for these sources. Therefore, given active sources, we can divide $b_g$ components into multiple groups, then design CD scheduling intervals for each group independently. 
 \item The policy assigns CD scheduling intervals to sources in $\bigcup_{i=1}^{3}\bm{\mathcal{V}}^{(i)}$ instead of other combinations of components. This is because extra wireless resources paid when scheduling intervals are decreased for a CD structure are minimized in this case. To illustrate this point, we use the sum of updating rates of all sources to estimate the minimum number of required channels. In Example 2, the sum of updating rates is $1.875$. If we require CD scheduling intervals for other combinations of components, for example $\{\bm{\mathcal{V}}^{(2)}$,$\bm{\mathcal{V}}^{(3)}$,$\bm{\mathcal{V}}^{(4)}\}$, then the minimum sum of updating rates is $2.1$, with $c_C, c_D, c_E=5$. 
 \end{itemize}
 
 \subsection {Scheduling Active Sources}
 
 Inspired by Example 2, given active sources, we divide $b_g$ components into multiple groups and design CD scheduling intervals for each group. Then we choose offsets such that the number of required channels is minimized. 
 Due to $T_n$, the mathematical relationship between CD scheduling intervals and the number of required channels remains unclear. Therefore, we use the minimum positive integer no less than the sum of updating rates of sources to approximate the number of channels required by a certain group.  
 The problem can be formulated as follows:
\begin{equation}\label{grouping}
	\begin{split}
		\min_{\bm{\mathcal{P}},\bm{c}} & \sum_{j=1}^{|\bm{\mathcal{P}}|}\left\lceil{ \sum_{i\in \bm{g}_j}\frac{1}{c_i}}\right\rceil,
		\\
		s.t.&  \left\{\begin{array}{lc}
			1\leq c_m\leq \frac{1}{l_m^{\prime}}, \forall m=1,2,\dots,M,\\
			\bm{c}_j \text{ is CD}, \forall j=1,\dots,|\bm{\mathcal{P}}|,\\
		    \text{There exists } j_i \text { for each $i$ such that } \bm{\mathcal{V}}^{(i)}\nsubseteq\bm{g}_{j_i},
		\end{array}\right.
	\end{split}
\end{equation}
where $\bm{\mathcal{P}}$ is a grouping method of $b_g$ components. $|\bm{\mathcal{P}}|$ is the number of groups. Sources in group $j$ is denoted by a set $\bm{g}_j$ and the corresponding scheduling intervals are denoted by $\bm{c}_j$. In the first constraint, $c_m$ is the scheduling interval of source $m$ and $\frac{1}{l_m^{\prime}}$ is the maximum scheduling interval derived by solving \eqref{active source}. The last two constraints mean that sources in the same component should be put in the same group and each group requires CD scheduling intervals. 

 Problem \eqref{grouping} is hard to solve due to the large solution space. Therefore, a heuristic algorithm named distance-based clustering (DC) is proposed in Algorithm 1. First of all, we determine the optional number of groups (Line 1). Since $\bm{\mathcal{P}}$ is a partition of the components, $|\bm{\mathcal{P}}|\leq b_g$. Additionally, we assume the base of scheduling intervals for each group possesses different values. Thus there are no more than $\min\{b_g,|\bm{u}|\}$ groups. Here $\bm{u}$ includes all distinct values of set $\{\frac{1}{l_m^{\prime}}|m\in\bigcup_{n=1}^{N}\bm{\mathcal{M}}_n\}$ and $|\bm{u}|$ is the number of elements of $\bm{u}$. For each number of groups, we identify all possible combinations of bases (Line 2). For each combination of bases (Line 3), DC puts sources in $\bm{\mathcal{V}}^{(j)}$ into the `nearest' group (Line 4 to 8).  
 Distance in the algorithm is defined as follows:
 
 \begin{definition} \label{distance}
 	Distance between the $i$-th group with base $e$ and $\bm{\mathcal{V}}^{(j)}$ is:
 	\begin{equation}
     		D_{i,j}=\sum_{m\in\bm{\mathcal{V}}^{(j)}}h_{i,m},
 	\end{equation}
 where: 
 	\begin{equation}
 		\begin{split}
 			h_{i,m}=&  \left\{\begin{array}{lc}
 				\frac{1}{\lfloor\frac{1}{l_m^{\prime}}/e\rfloor e}-l_m^{\prime}, &\frac{1}{l_m^{\prime}}\geq e;\\
 				\infty, &\frac{1}{l_m^{\prime}}<e.\\	
 			\end{array}\right.
 		\end{split}
 	\end{equation}	  
 \end{definition}

 According to the definition, the policy puts sources in $\bm{\mathcal{V}}^{(j)}$ into the group with a base close to divisors of $\frac{1}{l_m^{\prime}},m\in\bm{\mathcal{V}}^{(j)}$. This is because we aim at minimizing extra wireless resources required when we decrease the scheduling intervals of some sources for a CD structure. Since we require CD scheduling intervals for each group, sources with $\frac{1}{l_m^{\prime}}$ smaller than the base can not be packed in the corresponding group. Therefore, if the current combination of bases leads to infinite distance between a certain source to all groups, we will abandon such a combination of bases and move to the next one (Lines 5 to 7). After determining a grouping scheme given bases of groups, we then find CD scheduling intervals for each group with the minimum sum of updating rates (Line 8) by a modified version of Aion which is a fast procedure proposed in \cite{liu2021aion}.

 \begin{algorithm}[t]
 	\caption{Distance-based Clustering (DC)}
 	\begin{algorithmic}[1]
 		\REQUIRE  
 		Active sources and the corresponding maximum scheduling intervals;
 		\ENSURE 
 		$\bm{\mathcal{P}}$,$\bm{c}^{\star}$; 
 		\FOR{$i\in[1,\min\{b_g,|\bm{u}|\}]$}
 		\STATE The optional set of centers: $\bm{u}_{ce}=\{\bm{z}|\bm{z}\nsubseteq\bm{u},|\bm{z}|=i\}$.
 		\FOR {each optional set of centers}
 		\IF {$\min_{j\leq i}h_{i,m}<\infty,\forall m\in\bigcup_{n=1}^{N}\bm{\mathcal{M}}_n$}
 		\FOR {$j=\{1,\cdots,b_g$\}}
 		\STATE Put sources in $\bm{\mathcal{V}}^{(j)}$ into the group with the minimum distance.
 		\ENDFOR
        \STATE Derive CD scheduling intervals with minimum sum of updating rates for each group.  
        \ENDIF
 		\ENDFOR
 		\ENDFOR
 		\STATE $\bm{\mathcal{P}}$ is a grouping method minimizing the objective of \eqref{grouping} over all the optional number of groups and bases. $\bm{c}^{\star}$ is the corresponding scheduling intervals.
 		
 		\label{code:recentEn}
 	\end{algorithmic}
 \end{algorithm}

 Given the grouping scheme and corresponding scheduling intervals $\bm{c}^{\star}$, we then choose offsets such that the number of required channels is minimized. 
 Such offsets $\bm{o}=\{o_m|m\in\bigcup_{n=1}^{N}\bm{\mathcal{M}}_n\}$ can be derived by solving the following problem:
 \begin{equation}\label{solve offset}
 	\begin{split}
 		\min_{\bm{o}} & \max_{t=1,\cdots,C}\sum_{m\in\bigcup_{n=1}^N\bm{\mathcal{M}}_n}U_m(t),
 		\\
 		s.t.\quad&  \left\{\begin{array}{lc}
 			\max_{m\in\bm{\mathcal{M}}_n}o_{m}\leq o_{z_n}, \forall n,\\
 			\max_{m\in\bm{\mathcal{M}}_n}(o_{z_n}-o_{m})\;\mathrm{mod}c_m^{\star}\leq T_n, \forall n,\\
 			o_m\leq c_m^{\star},\\
 			U_{m}(o_m+jc_{m}^{\star})=1, \forall
 		 j\in\{0,1,\cdots,\frac{C}{c_m^{\star}}-1\},\\
 		\end{array}\right.
 	\end{split}
 \end{equation}
 where $C$ is the least common multiple of elements in $\bm{c}^{\star}$ and $z_n$ is the source with the largest scheduling interval $c_{z_n}^{\star}$ in $\bm{\mathcal{M}}_n$. Since we only consider homogeneous scheduling policies with fixed scheduling interval $c_m$ and offset $o_m\leq c_m$, $U_m(t+C)=U_m(t)$ holds for all time slots and all sources. Therefore, we only need to solve \eqref{solve offset} within any $C$ consecutive time slots. The objective is the number of required channels. Constraints in \eqref{solve offset} ensure that each scheduling of source $z_n$ leads to a successful fusion and an AoI reduction in the next slot. 
 
 Based on previous analysis, SCPA is proposed in Algorithm 2. Given AoI constraints for $N$ regions, we first identify active sources and the maximum scheduling interval for each source (Line 1). Then we divide active sources into several groups and design CD scheduling intervals for sources in each group with DC (Line 2). Finally, we choose offsets for sources such that the number of required channels is minimized (Line 3).

\section{Numerical Results}

\subsection{Setup}

We consider 2 scenarios, the first one includes $3\times3=9$ regions and the second one has $6\times6=36$ regions.  
Regions monitored by each source are determined by three factors (Fig.~\ref{three factors}): 1) Location of the source; 2) Orientation of the source, there are 4 possible choices: up, down, left, and right. 3) Coverage of the source, there are 3 possible options: one region, 2 adjacent regions, and 3 adjacent regions.  
We assume $\bm{\mathcal{F}}_n$ includes all sources located in region $n$ and each $\bm{k}_{n,i}$ includes 2 different sources that cover region $n$ but are not located within it. 
To derive the following results, we set a source located in each region. 

\begin{algorithm}[t]
	\caption{Scheduling for CP with Asynchronous status updates (SCPA)}
	\begin{algorithmic}[1]
		\REQUIRE  
		$\bm{d}$,  $\{\bm{\mathcal{K}}_1, \cdots, \bm{\mathcal{K}}_N\}$, $\{T_1,\cdots,T_N\}$;
		\ENSURE 
		A homogeneous scheduling policy $\bm{\pi}$. 
		\STATE Determine active sources and the corresponding maximum scheduling intervals by solving \eqref{active source}.
		\STATE Determine scheduling intervals $\bm{c}^{\star}$ for the sources by DC.  
		\STATE Determine offsets $\bm{o}$ for the sources by solving \eqref{solve offset}.
	\end{algorithmic}
\end{algorithm}

\begin{figure} 
	\centering
	\subfigure[]{\includegraphics[scale=0.29]{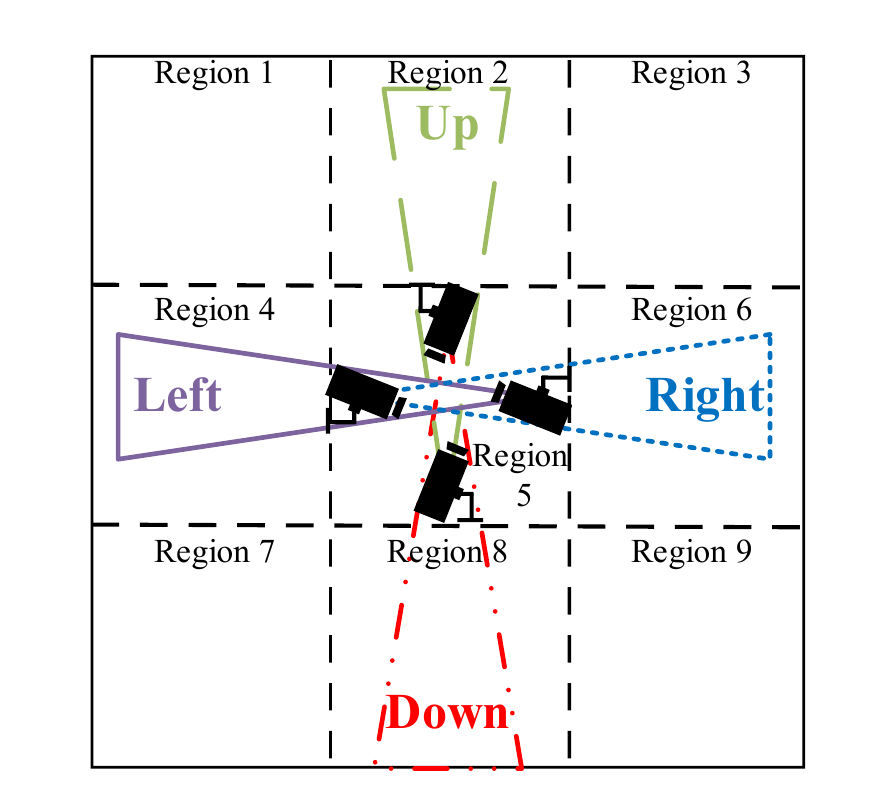}}
	\subfigure[]{\includegraphics[scale=0.29]{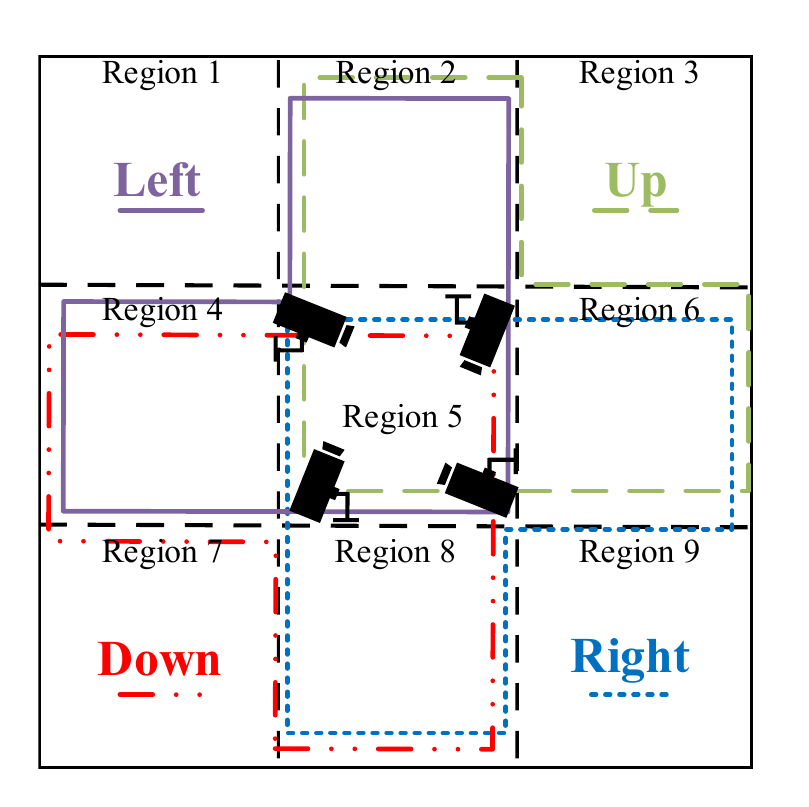}}
	\caption{Regions covered by a source with different orientation and same location, when the coverage of source is (a)  2 regions; (b) 3 regions.}
	\label{three factors}
\end{figure}

\begin{figure}
	\centering
	\includegraphics[scale=0.35]{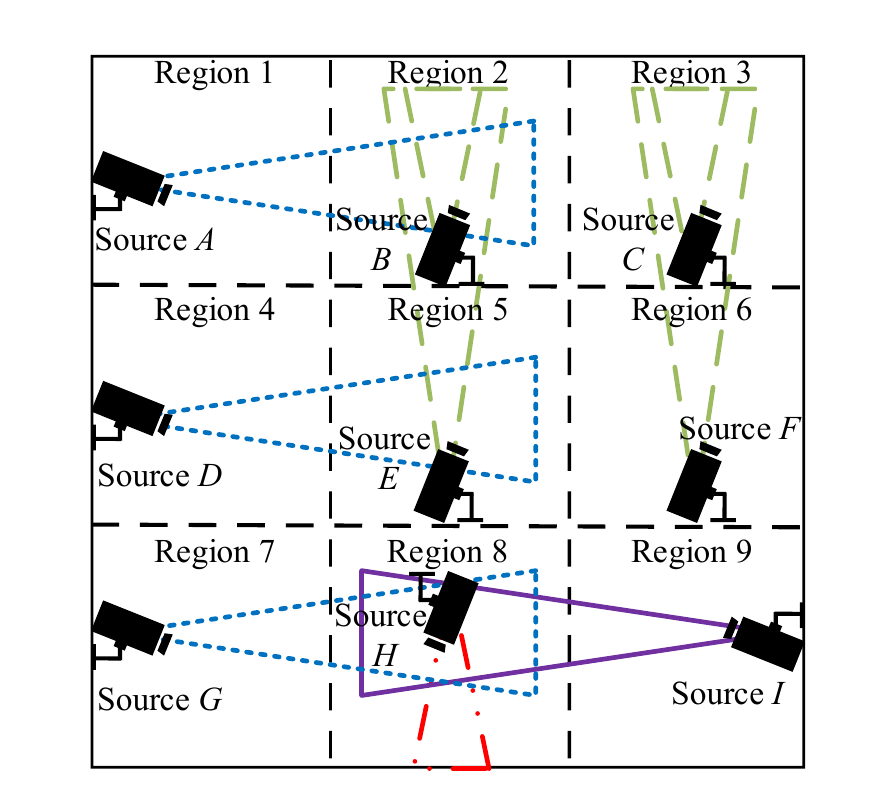}
	\vspace{-0.1in}
	\caption{An example with $3\times3$ regions and $9$ sources.}
	\label{r9example}
\end{figure}

\subsection{An example}

We start with an example to illustrate SCPA and show AoI constraints are satisfied. Consider $9$ regions and $9$ sources with coverage of 2 regions from $A$ to $I$ in Fig.~\ref{r9example}. However, each source in $\{B,C,H\}$ covers a single region due to the orientation. AoI constraints are $[6, 5, 2,  2, 7, 4, 3, 8, 7]$ and $T_n=1,\forall n$. By solving \eqref{active source}, active sources are $\{A,C,D,E,F,G,I\}$ and the corresponding maximum scheduling intervals are $[5,2,2,5,4,3,7]$. All regions possess an active source in $\bm{\mathcal{F}}_n$, except regions 2 and 8. Active sources for these two regions are $\{A,E\}$ and $\{G,I\}$. If DC packs all sources in the same group, then the corresponding scheduling intervals are $[4,2,2,4,4,2,4]$, possessing CD structure after sorting. After that, we obtain offsets $[3,2,2,3,1,1,4]$ for active sources by solving \eqref{solve offset}. Then the homogeneous scheduling policy is:
\vspace{-0.1in}
\begin{center}
	\setlength{\tabcolsep}{1.3mm}{
		\begin{tabular}{l c c c c c c c c c c c c c} 		
			time slot&1&2&3&4&5&6&7&8&9&10&11&12\\
			channel 1 & F & D & A & D& F & D & A & D & F & D & A & D & $\cdots$	\\
			channel 2 &G&C& E & C & G &C& E & C & G & C & E & C & $\cdots$\\
			channel 3 & $\square$ & $\square$ & G & I& $\square$ & $\square$ & G & I & $\square$ & $\square$ & G & I & $\cdots$	\\
	\end{tabular}}
\end{center}
Under the policy, sources activated for nine regions except regions 2 and 8 are scheduled with intervals smaller than the corresponding AoI constraints. $A_2(t)$ and $A_8(t)$ reduce to 1 at least once in every 4 slots. The constraints are satisfied. 

\subsection{Performance of the proposed algorithms}

In this subsection, we compare the number of channels required by the proposed scheduling policies and the lower bound in Lemma \ref{lbub}. Simulations are carried in the second scenario. Three different cases are considered:
\begin{itemize}
\item Case 1: the center employs fusion algorithms with $T_n=d_n-1, \forall n$.
\item Case 2: the center employs fusion algorithms with $T_n=1,\forall n$.
\item Case 3: the center is not able to fuse status updates from multiple sources and can only rely on sources in $\bm{\mathcal{F}}_n$ to obtain information of region $n$. 
\end{itemize}

Fig.~\ref{lbr36} displays the lower bounds and the number of channels required by SCPA in the above three cases with different coverage of sources. 
The lower bound in Lemma 1 is independent with $T_n$ and the same bar is used to show the bound in Cases 1 and 2. The coverage of sources are the same and we increase the coverage from a single region to $3$ regions with fixed location and orientation of sources. We generate different orientations of sources and AoI constraints of $N$ regions uniformly  $50$ times and only show the averaged results.

Additionally, if $T_n=d_n-1$, then \eqref{offsets_general scheduling interval} always hold given whatever offsets if scheduling intervals of sources in $\bm{\mathcal{M}}_n$ do not exceed $d_n$. Therefore, we can use TGA proposed in \cite{wang2022grouping} for active sources given their maximum scheduling intervals. Since TGA can produce not only homogeneous policies and achieve near-optimal performance, we use it after SCPA identifies active sources in Case 1. Similarly, since the center is not able to fuse information in Case 3, we also determine active sources by SCPA and use TGA to establish a policy for them. To reduce the running time of SCPA, we limit the optional number of groups to no more than 2 and set the only optional combination of bases to be $\{2,3\}$ when there are 2 groups. Since $\bm{\mathcal{F}}_n\neq\emptyset, \forall n$, we assume the center will schedule sources under the scheduling policy established for Case 3 if the policy for Case 1 or 2 requires more channels.

\begin{figure}
	\centering
	\includegraphics[scale=0.45]{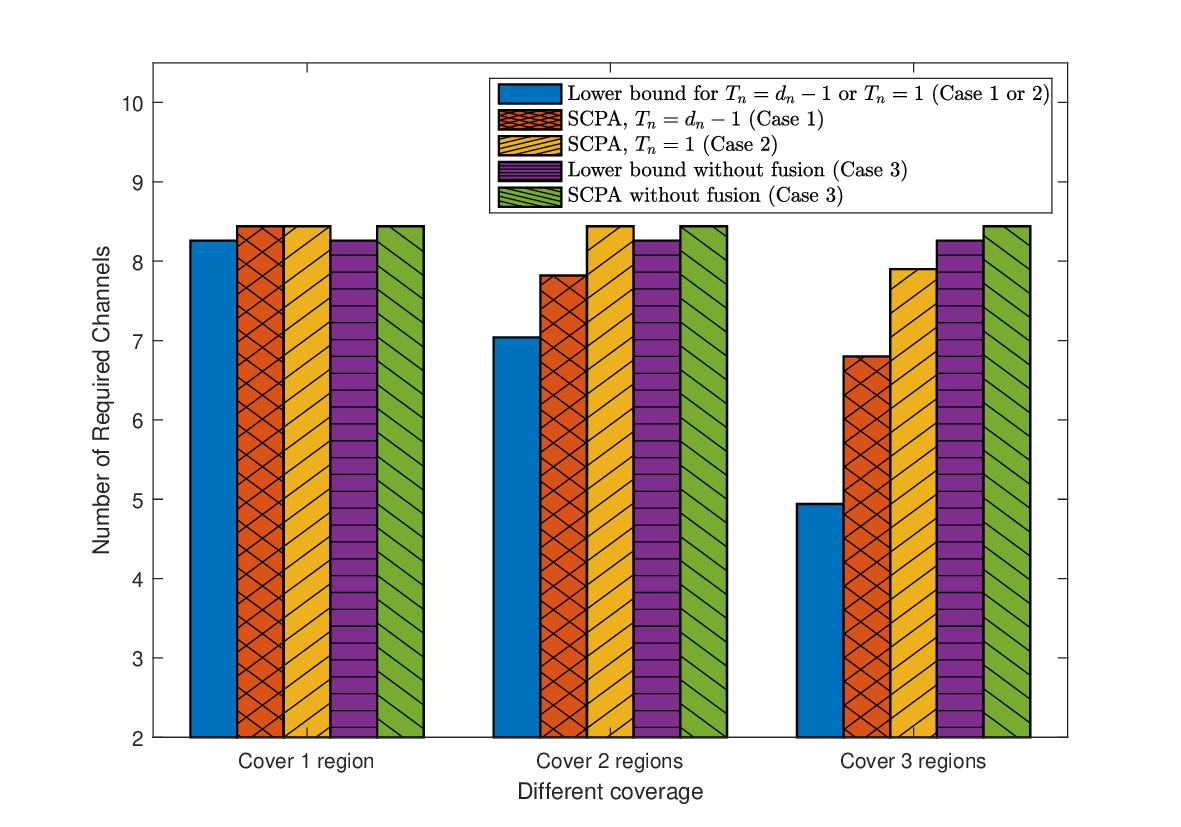}
	\vspace{-0.2in}
	\caption{Lower bound and performance of SCPA over 50 instances, where AoI
		constraints are generated uniformly in interval [2, 10].}
	\label{lbr36}
\end{figure}

In Fig.~\ref{lbr36}, it can be seen that large coverage of sources reduces the number of required channels if the center employs fusion algorithms. When each source covers a single region, we need information from 36 sources and the lower bounds of the three cases are the same. This is because status updates from different sources include information of different regions and can not be used for fusion. When the coverage increases to 2 or 3 regions, the lower bounds for the first two cases decrease. This is because each status update may include information of several regions and lead to AoI reduction of these regions with fusion algorithms. In this case, the number of required channels will decrease.  

Similar to the lower bounds, the numbers of channels required by SCPA are all the same in three cases when coverage is a single region. When coverage increases to 2 regions, the gap between the lower bound and the number of channels required by SCPA is $11.08\%$ and $19.89\%$ in the first two cases, respectively. Note that the number of channels required in Cases 2 and 3 remains the same, even when the number of activated sources is smaller in Case 2. This is because offsets of active sources are constrained by $T_n=1$. 
When the coverage increases to 3 regions, the advantage of information fusion is expanded. Compared with the number of channels required in Case 3, the numbers of channels required in the first two cases are reduced by $19.43\%$ and $6.4\%$, respectively. This is because the number of activated sources is further reduced. 
 Also, the gap between the lower bound and the number of required channels increases in Cases 1 and 2.  
In this case, 
reconstruction of information of more regions will rely on sources in $\overline{\bm{\mathcal{F}}}_n$. Meanwhile, SCPA requires CD scheduling intervals of sources in $\bm{\mathcal{M}}_n$. Therefore, scheduling intervals and offsets for active sources will be more constrained when the coverage increases. Finally, Case 1 requires the minimum number of channels. This shows that the development of fusion algorithms can reduce the communication load.

\section{Conclusion}

We focus on a novel scheduling problem under hard AoI constraints in CP with asynchronous status updates. An algorithm named SCPA is proposed to establish a scheduling policy and minimize the number of required channels. Since there may be overlapping among status updates from multiple sources, SCPA first chooses active sources with the minimum sum of updating rates. SCPA then requires sources activated for the same region to be scheduled with CD intervals. Such an operation ensures the existence of a homogeneous scheduling policy satisfying AoI constraints. Given CD scheduling intervals, SCPA identifies offset for each active source with the minimized number of required channels. The number of channels required by SCPA can reach only 12\% more than the derived lower bound according to numerical results.

\section{Acknowledgment}
This work is sponsored in part by the National Key R\&D Program of China No. 2020YFB1806605, by the Natural Science Foundation of China (No. 62341108, No. 62221001, No. 62022049, No. 62111530197), by the Beijing Natural Science Foundation under grant L222044, by the Talent Fund of Beijing Jiaotong University under grant 2023XKRC030 and Hitachi Ltd.

\bibliography{test_mosc}
\bibliographystyle{IEEEtran}

\end{document}